\documentclass[12pt]{article}%
\usepackage{amssymb}
\usepackage{amsfonts}
\usepackage{amsmath}
\usepackage{graphicx}%
\providecommand{\U}[1]{\protect\rule{.1in}{.1in}}
\newtheorem{theorem}{Theorem}

\newtheorem{lemma}{Lemma}

\newenvironment{proof}[1][Proof]{\noindent\textbf{#1.} }{\ \rule{0.5em}{0.5em}}

\begin{document}

\title{Conditional probability calculation using restricted Boltzmann machine with
application to system identification}
\author{Erick de la Rosa, Wen Yu\\{\small Departamento de Control Automatico}\\{\small CINVESTAV-IPN (National Polytechnic Institute)}\\{\small Mexico City, 07360, Mexico.\ yuw@ctrl.cinvestav.mx}}
\date{}
\maketitle

\begin{abstract}
There are many advantages to use probability method for nonlinear system
identification, such as the noises and outliers in the data set do not affect
the probability models significantly; the input features can be extracted in
probability forms. The biggest obstacle of the probability model is the
probability distributions are not easy to be obtained.

In this paper, we form the nonlinear system identification into solving the
conditional probability. Then we modify the restricted Boltzmann machine
(RBM), such that the joint probability, input distribution, and the
conditional probability can be calculated by the RBM training. Binary encoding
and continue valued methods are discussed. The universal approximation
analysis for the conditional probability based modelling is proposed. We use
two benchmark nonlinear systems to compare our probability modelling method
with the other black-box modeling methods. The results show that this novel
method is much better when there are big noises and the system dynamics are complex.

\textit{Keywords: system identification, conditional probability, restricted
Boltzmann machine}

\end{abstract}

\section{Introduction}

Data based system identification is to use the experimental data from system
input and output. Usually, the mathematical model of its input-output behavior
is applied to predict the system response (output) from the excitation (input)
to the system. There is always uncertainty when a model is chosen to represent
the system. Including probability theory in dynamic system identification can
improve the modeling capability with respect to noise and uncertainties
\cite{Ljung}. The common probability approaches for system identification
parameterize the model and apply Bayes' Theorem \cite{Pillonetto}, such as
maximizing the posterior \cite{Juloski}, maximizing the likelihood function
\cite{Schona}, or matching the output by least-squares method \cite{Lu}. There
are several computational difficulties with the above estimation methods, for
example, the probability models for the likelihood construction cannot be
expected to be perfect; the parameter estimations are often not unique,
because the true values of the parameters may not exist.

In the sense of probability theory, the objective of system modeling is to
find the best conditional probability distribution $P(y|\mathbf{x})$
\cite{Erhan1}, where $\mathbf{x}$ is the input and $y$ is the output. So the
system identification can be transformed to the calculation of the conditional
probability distribution, it is no longer the parameter estimation problem.
For nonlinear system identification, there are two correlated time series, the
input $\mathbf{x}$ and the output $y$. The popular Monte Carlo method cannot
obtain the best prediction of $y$ from the conditional probability between of
$\mathbf{x}$ and $y$ \cite{Damien}.

Recent results, show that restricted Boltzmann machine (RBM) \cite{Fisher} can
learn the probability distribution of the input data using the unsupervised
learning method, and obtains their hidden features \cite{Hinton}\cite{Bengio}.
These latent features help the RBM to obtain better representations of the
empirical distribution. Feature extractions and pre-training are two important
properties of the RBMs for solving classification problems in the past years
\cite{Hinton1}.

RBMs are also applied for data regression and time series modeling. With RBM
pre-training, the prediction accuracy of the time series can be improved
\cite{Zhang}\cite{Busseti}. The hidden and visible units of normal RBM are
binary. The prediction accuracy of the time series with continuous values are
not satisfied \cite{Langgkvist}. In \cite{Nair}, the binary units are replaced
by linear units with Gaussian noise. The denoising autoencoders are used for
continuous valued data in \cite{Romeuu}. In \cite{Shang}, the time series are
assumed to have the Gaussian property.

In order to find the relation between the input $\mathbf{x}$ and the output
$y$, the RBM is used to obtain the conditional probabilities between
$\mathbf{x}$ and the hidden units of the RBM in our previous papers
\cite{delarosa1}\cite{delarosa2}. The hidden units are used as the initial
weights of neural networks, then the supervised learning is implemented to
obtain the inpu-output relation, $y=f\left(  \mathbf{x}\right)  $. To the best
of our knowledge, conditional probability approach for system identification
has not been still applied in the literature.

In this paper, we use the conditional probability to model the nonlinear
system by RBM training. The RBMs are modified, such that the input
distribution, the joint probability $P(y,x),$ and the conditional probability
$P(y|x)$ can be calculated. We proposed two probability calculation methods:
binary encoding and continuous values. Probability gradient algorithm is used
to maximize the log-likelihood of the conditional probability between the
input and output. The comparisons with the other black-box identification
methods are carried out using two nonlinear benchmark systems.

\section{Nonlinear system modeling with conditional probability}

We use the following difference equation to describe a discrete-time nonlinear
system,%
\begin{equation}
y(k)=f\left[  \mathbf{x}\left(  k\right)  ,\xi\left(  k\right)  ,\cdots
\xi\left(  k-n_{\xi}\right)  \right]  \label{planta}%
\end{equation}
where $\xi\left(  k\right)  $ are noise sequences, $f\left(  \cdot\right)  $
is an unknown nonlinear function,
\begin{equation}
\mathbf{x}\left(  k\right)  =[y\left(  k-1\right)  ,\cdots y\left(
k-n_{y}\right)  ,u\left(  k\right)  ,\cdots u\left(  k-n_{u}\right)  ]^{T}
\label{X}%
\end{equation}
representing the plant dynamics, $u\left(  k\right)  $ and $y\left(  k\right)
$ are the measurable input and output of the nonlinear plant, $n_{y}$ and
$n_{u}$ correspond to the system order, $n_{\xi}$ is the maximum lag for the
noise. $\mathbf{x}\left(  k\right)  \in\Re^{n}$ can be regarded as a new input
to the nonlinear function $f\left(  \cdot\right)  .$ It is the well known
NARMAX model \cite{Billings}.

The objective of the system modeling is to use the input-output data, to
construct a model $\hat{y}(k)=\hat{f}\left[  \mathbf{x}\left(  k\right)
\right]  ,$ such that $\hat{y}(k)\rightarrow y(k).$ If the lost function is
defined by%
\[
L=\sum_{k}\left[  y(k)-\hat{f}\left[  \mathbf{x}\left(  k\right)  \right]
\right]  ^{2}%
\]
The mathematical expectation of the modeling error is%
\begin{equation}
E\left\{  \left[  y(k)-\hat{f}\left[  \mathbf{x}\left(  k\right)  \right]
\right]  ^{2}\right\}  =\int\left[  y-\hat{f}\left(  x\right)  \right]
^{2}p\left(  dx,dy\right)  \label{me1}%
\end{equation}
where the joint probability satisfies
\[
p\left(  X,Y\right)  =p\left(  X\right)  p\left(  Y\mid X\right)
\]
Here $p\left(  Y\mid X\right)  $ is the conditional probability. The modeling
error (\ref{me1}) becomes%
\[%
\begin{array}
[c]{c}%
E\left\{  \left[  y(k)-\hat{f}\left[  \mathbf{x}\left(  k\right)  \right]
\right]  ^{2}\right\}  =\int\left[  y-\hat{f}\right]  ^{2}p\left(  dx\right)
p\left(  dy\mid dx\right) \\
=E_{\mathbf{x}}E_{\left(  y\mid\mathbf{x}\right)  }\left\{  \left[  y-\hat
{f}\right]  ^{2}\mid\mathbf{x}\right\}
\end{array}
\]
The system identification becomes the minimization problem as%
\[
\min E\left\{  \left[  y-\hat{f}\right]  ^{2}\right\}  \rightarrow\min
_{\hat{f}}E_{\left(  y\mid\mathbf{x}\right)  }\left\{  \left[  y-\hat
{f}\right]  ^{2}\mid\mathbf{x}\right\}
\]
The best prediction of $y$ at $\mathbf{x}=x$ is the conditional mean
(conditional expectation) as%
\[
\hat{f}\left[  \mathbf{x}\left(  k\right)  \right]  =E\left\{  \hat{y}%
\mid\mathbf{x}\right\}
\]
The nonlinear system modeling becomes
\begin{equation}
\max_{\theta}\left\{  p\left(  y\mid\mathbf{x},\Lambda\right)  \right\}
\label{modeling}%
\end{equation}
where $\Lambda$ is the parameter vector of the model $\hat{f}.$ The existing
methods use some probability approaches, such as Bayes' Theorem, to estimate
$\Lambda$. In this paper, we do not use parameterized models. We will
calculate the conditional probability distribution $p\left[  y\left(
k\right)  |\mathbf{x}\left(  k\right)  \right]  $ directly.

The loss function for the conditional distribution $p(y|\mathbf{x})$ is
defined as
\begin{equation}
J_{o}(D)=\sum_{D}\log p\left(  y|\mathbf{x}\right)  \label{logJ}%
\end{equation}
where $D=\{\mathbf{x}(k),y(k)\},$ is the training set, $\mathbf{x}(k)$ and
$y(k)$ are the $k$-th training input vector and output respectively. So the
object of the nonlinear system identification (\ref{modeling}) becomes
\begin{equation}
\max_{\Lambda}\left[  \sum_{D}\log p\left(  y|\mathbf{x}\right)  \right]
\label{maxJ}%
\end{equation}

In this paper, we use the restricted Boltzmann machine (RBM) to obtain the
best conditional probability. RBM is a stochastic neural network. It can learn
the probability distribution of given data set. The input (or the visible
nodes) to the RBM is $\mathbf{x}\left(  k\right)  =\left[  x_{1}\cdots
x_{n}\right]  \in R^{n}$. The output (or the hidden nodes) of the RBM is
$h=\left[  \bar{h}_{1}\cdots\bar{h}_{s}\right]  \in R^{s}$. For the $i-th$
hidden node and the $j-th$ visible node, the conditional probabilities are
calculated as%
\begin{equation}%
\begin{array}
[c]{l}%
p\left(  \bar{h}_{i}=1\mid\mathbf{x}\right)  =\phi\left[  W\mathbf{x}+b\right]
\\
p\left(  x_{j}=1\mid h\right)  =\phi\left[  W^{T}h+c\right]
\end{array}
\label{rbm1}%
\end{equation}
where $\bar{h}_{i}=\left\{
\begin{array}
[c]{cc}%
1 & a<p\left(  \bar{h}_{i}=1\mid\mathbf{x}\right) \\
0 & a\geq p\left(  \bar{h}_{i}=1\mid\mathbf{x}\right)
\end{array}
\right.  $, $\phi$ is the Sigmoid function, $W$ is the weight matrix, $a$ is a
number sampled from the uniform distribution over $[0,1]$, $b$ and $c$ are
visible and hidden biases respectively, $i=1,\ldots,s$, $j=1,\ldots,n.$ The
probability vector $h$ is defined as
\[
h=\left[  p\left(  \bar{h}_{1}=1\mid\mathbf{x}\right)  \cdots p\left(  \bar
{h}_{s}=1\mid\mathbf{x}\right)  \right]  =\left[  h_{1}\cdots h_{s}\right]
\]

The training object of the RBM is to modify the parameters $\left[
W,b,c\right]  ,$ such that the probability distribution of the hidden units
$p$ is near to the distribution of the input $q$. We use Kullback-Liebler (KL)
divergence to measure the distance between two probability distributions $p$
and
\begin{equation}
KL\left(  p,q\right)  =\sum_{D}q\left(  \mathbf{x}\right)  \log\frac{q\left(
\mathbf{x}\right)  }{p\left(  \mathbf{x}\right)  }=\sum_{D}q\left(  x\right)
\log q\left(  x\right)  -\sum_{D}q\left(  x\right)  \log p\left(  x\right)
\label{kl}%
\end{equation}
Here the first term $\sum_{x}q\left(  x\right)  \log q\left(  x\right)  $ is
the entropy of the input. It is independent of the RBM training.

The RBM training object becomes
\begin{equation}
\min KL\left(  p,q\right)  \rightarrow\max\sum_{D}q\left(  x\right)  \log
p\left(  x\right)  \label{rbmO}%
\end{equation}
$q\left(  x\right)  $ cannot be obtained directly, it is estimated by Monte
Carlo method%
\[
\sum_{D}q\left(  x\right)  \log p\left(  x\right)  \approx\frac{1}{n}\sum
_{D}\log p\left(  x\right)
\]
where $n$ is the number of training data. So%
\[
\min KL\left(  p,q\right)  \rightarrow\max\sum_{x}\log p\left(  x\right)
\]

The following lemma and theorem give the universal approximation of the
nonlinear system (\ref{planta}) with the conditional probability
$p(y|\mathbf{x})$ and the RBM.

\begin{lemma}
Any marginal probability distribution $p(\mathbf{x}),$ $\mathbf{x}%
\in\{0,1\}^{r}$, can be approximated arbitrarily well in the sense of the KL
divergence\ (\ref{kl}) by then RBM with $r+1$ hidden units, where $r$ is the
number of input vector whose probability is not $0$ \cite{Leroux}.
\end{lemma}

\begin{theorem}
If the hidden units of the RBM (\ref{rbm1}) is big enough, the conditional
probability $p(y|\mathbf{x})$ of the nonlinear system (\ref{planta}) can be
approximated arbitrarily well by an RBM and the pair $(\mathbf{x},y),$ in the
sense of the KL divergence\ (\ref{kl}).
\end{theorem}

\begin{proof}
Consider the nonlinear system (\ref{planta}), the vectors $\mathbf{x}$ and $y$
are from the finite sets $\left\{  \mathbf{x}_{1},...,\mathbf{x}%
_{k,}...,\mathbf{x}_{r_{0}}\right\}  $ and $\left\{  y_{1},...,y_{l,}%
...,y_{m_{0}}\right\}  $. For each pair $\left(  \mathbf{x}_{k},y_{l}\right)
$, the conditional probability distribution is $p(y_{l}|\mathbf{x}_{k}),$
$k=1,...,r_{0}$ and $l=1,...m_{0}$. The conditional probability of the pair is%
\begin{equation}
p(y_{l}|\mathbf{x}_{k})=\frac{p(\mathbf{x}_{k},y_{l})}{p(\mathbf{x}_{k}%
)}=\frac{p(\mathbf{x}_{k},y_{l})}{\sum_{j}p(\mathbf{x}_{k},y_{j})}
\label{conGen}%
\end{equation}
Here the value of $p(y_{l}|\mathbf{x}_{k})$ is assumed to be known. The term
$\sum_{j}p(\mathbf{x}_{k},y_{j})$ can be separated as $\sum_{j\neq
l}p(\mathbf{x}_{k},y_{j})$ and $p(\mathbf{x}_{k},y_{l}).$ From (\ref{conGen}),%
\begin{equation}
p(\mathbf{x}_{k},y_{l})=\frac{p(y_{l}|\mathbf{x}_{k})\sum_{j\neq
l}p(\mathbf{x}_{k},y_{j})}{1-p(y_{l}|\mathbf{x}_{k})} \label{ecSis}%
\end{equation}
For each pair $\left(  \mathbf{x}_{k},y_{l}\right)  $, $p(\mathbf{x}_{k}%
,y_{l})$ can be calculated for the indexes $k$ and $l$ from (\ref{ecSis}).
There are $r_{0}\times m_{0}$ equations. They are considered as the desired
conditional distribution of $p(y|\mathbf{x})$, \textit{i.e.}, we can use
$p(\mathbf{x}_{k},y_{l})$ to create the distribution which contains the
original conditional distribution $p(y|\mathbf{x})$. Now we apply all pairs
$\left(  \mathbf{x}_{k},y_{l}\right)  ,$ $k=1,...,r_{0}$ and $l=1,...m_{0},$
to the RBM whose hidden unit number is $r_{0}\times m_{0}+1.$

If $\mathbf{x}_{k}$ and $y_{l}$ are binary values, $\mathbf{x}\in\{0,1\}^{m},$
$y\in\{0,1\}^{r},$ we define the pair $(\mathbf{x}_{k},y_{l})$ as a single
random variable $z_{kl},$ $p(z_{kl})=p(\mathbf{x}_{k},y_{l}).$ From Lemma 1,
we can construct an RBM with the most $r_{0}\times m_{0}+1$ hidden units (all
input nodes have non-zero probability), which models the distribution
$p(z_{kl})$, with the desired conditional distribution of $p(y|\mathbf{x}).$

If $\mathbf{x}_{k}$ and $y_{l}$ are not binary values, we encode the input
variable $\mathbf{x}$ and $y$ into binary values $\left\{  0,1\right\}  $ with
the resolution of $m$ bits. This means that $y\left(  k\right)  $ is encoded
into $2^{m}$ different levels with the step of $1/(2^{m}-1).$ Similarly, the
control $\mathbf{x}\left(  k\right)  $ is encoded into $2^{m}$ different
levels,%
\begin{equation}%
\begin{array}
[c]{l}%
\mathbf{x}\in\Re^{n}\longrightarrow\mathbf{x}\in\{0,1\}^{n\times m}\\
y\in\Re\longrightarrow y\in\{0,1\}^{m}%
\end{array}
\label{bin}%
\end{equation}
We can construct an RBM with at most $\left(  n\times m\times r_{0}\right)
\times\left(  m\times m_{0}\right)  +1$ hidden units. The desired conditional
distribution of $p(y|\mathbf{x})$ can be approximated arbitrarily well in the
sense of the KL divergence.
\end{proof}

The above theorem can be regarded as the probability version of the universal
approximation theory of neural networks \cite{Cybenko}\cite{Huang}. One RBM
with one hidden layer can learn the probability $p(y|\mathbf{x})$ of the
nonlinear system in any accuracy, the hidden node number of the RBM should be
the same as total data number. This cause serious over-fitting and
computational problems.

In order to improve the approximation accuracy, we can use several RBMs with
cascade connection. It is called deep Boltzmann machines (or deep belief nets)
\cite{Hinton2}. Instead of using more hidden nodes, it use more layers (or
more RBMs) to learn the probability $p(y|\mathbf{x})$ to avoid the above
problems. This architecture is shown in Figure \ref{f1}. Each RBM has the form
of (\ref{rbm1}). The output of the current RBM is its hidden vector $h$. It is
the input of the next RBM.%

%TCIMACRO{\FRAME{ftbpFU}{3.7395in}{1.6933in}{0pt}{\Qcb{Deep Boltzmann machines
%for conditional probability of nonlinear system}}{\Qlb{f1}}{F1}%
%{\special{ language "Scientific Word";  type "GRAPHIC";  display "USEDEF";
%valid_file "T";  width 3.7395in;  height 1.6933in;  depth 0pt;
%original-width 6.1869in;  original-height 3.3788in;  cropleft "0";
%croptop "1";  cropright "1";  cropbottom "0";
%tempfilename 'fnn1.eps';tempfile-properties "XNPEUR";}} }%
%BeginExpansion
\begin{figure}
[ptb]
\begin{center}
\includegraphics[
height=1.6933in,
width=3.7395in
]%
{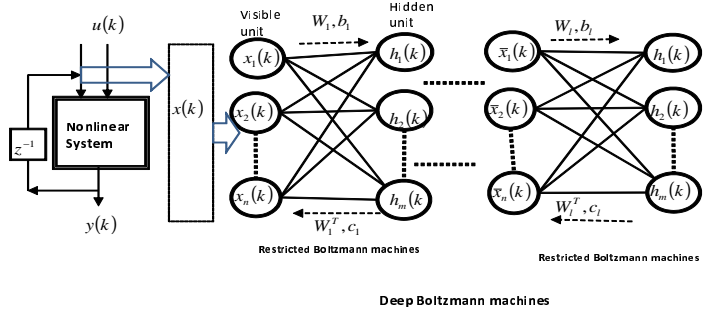}%
\caption{Deep Boltzmann machines for conditional probability of nonlinear
system}%
\label{f1}%
\end{center}
\end{figure}
%EndExpansion

\section{Joint distribution for nonlinear system identification}

Because the conditional probability%
\begin{equation}
p(y|\mathbf{x})=\frac{p(\mathbf{x,}y)}{p(\mathbf{x})},\quad\log p(y\mathbf{|x}%
)=\log p(\mathbf{x,}y)-\log p(\mathbf{x}) \label{rep}%
\end{equation}
The object of system identification (\ref{maxJ}) can also be formed as%
\begin{equation}
\max_{\Lambda}\left[  \sum_{D}\log p(\mathbf{x,}y)\right]  \text{ and }%
\min_{\Lambda}\left[  -\sum_{D}\log p(\mathbf{x})\right]  \label{maxJ1}%
\end{equation}

The identification object (\ref{maxJ1}) is an alternative method of the
conditional probability (\ref{maxJ}), which needs the joint distribution
$\sum_{D}\log p(\mathbf{x,}y)$ and the input distribution $\sum_{D}\log
p(\mathbf{x})$. It is more easy to calculate the joint distribution than the
conditional probability. However, it is impossible to optimize $\sum_{D}\log
p(\mathbf{x,}y)$ and $\sum_{D}\log p(\mathbf{x})$ with $D=\{\mathbf{x}%
(k),y(k)\}$ at same time. We can use part of the training data $\mathbf{x}$ to
minimize $\sum_{D}\log p(\mathbf{x})$ first, then use the rest of data to
maximize the joint probability $\sum_{D}\log p(\mathbf{x,}y)$.

In our previous works \cite{delarosa1}\cite{delarosa2}, we first use the input
data $\mathbf{x}$ to pre-train the RBM. The training results are used as the
initial weights of neural networks. Then we used the supervised learning to
train the neural model. In this paper, after the unsupervised learning for
$\sum_{D}\log p(\mathbf{x}),$ we train the RBM to obtain $\max_{\Lambda
}\left[  \sum_{D}\log p(\mathbf{x,}y)\right]  .$ For the nonlinear system
identification, its a sub-optimization process.

\subsection{Pre-training for $p(\mathbf{x})$}

The goal of the unsupervised training is to obtain $\min_{\Lambda}\left[
\sum_{D}\log p(\mathbf{x})\right]  $ by reconstructing the RBM$\mathbf{.}$ The
parameters $\Lambda=\left[  W,b,c\right]  $ are trained such that $h$ is the
representation of $\mathbf{x}$ (feature extraction). The probability
distribution $p\left(  \mathbf{x}\right)  $ is the following energy-based
model%
\begin{equation}
p(\mathbf{x})=\sum_{h}p(\mathbf{x},h)=\sum_{h}\frac{e^{-E(\mathbf{x},h)}}{Z}
\label{engfun}%
\end{equation}
where $Z=\sum_{h}\sum_{\mathbf{x}}e^{-E(\mathbf{x},h)}$ denotes the sums over
all possible values of $h$ and $\mathbf{x,}$ $E(x,h)$ is the energy function
which is defined by%
\begin{equation}
E\left(  \mathbf{x},h\right)  =-c^{T}\mathbf{x}-b^{T}h-h^{T}W\mathbf{x}
\label{eng}%
\end{equation}
The loss function for the training is
\[
L\left(  \Lambda\right)  =\log%
%TCIMACRO{\dprod \limits_{\mathbf{x}}}%
%BeginExpansion
{\displaystyle\prod\limits_{\mathbf{x}}}
%EndExpansion
p\left(  \mathbf{x}\right)  =\log\left[  \sum_{\mathbf{x}}e^{-E\left(
\mathbf{x},h\right)  }\right]  -\log\left[  \sum_{\mathbf{x},h}e^{-E\left(
\mathbf{x},h\right)  }\right]
\]
The weights are updated as
\begin{equation}
\Lambda\left(  k+1\right)  =\Lambda\left(  k\right)  -\eta_{1}\frac
{\partial\left[  -\log p\left(  \mathbf{x}\right)  \right]  }{\partial\Lambda}
\label{up3}%
\end{equation}
where $\eta_{1}$ is the learning rate, $\frac{\partial\log p\left(
\mathbf{x}\right)  }{\partial\Lambda}=\sum_{x}p\left(  \mathbf{x}\right)
\frac{\partial\digamma(\mathbf{x})}{\partial\Lambda}-\frac{\partial
\digamma(\mathbf{x})}{\partial\Lambda},$ $\digamma(\mathbf{x})$ is the free
energy defined as
\[
\digamma(\mathbf{x})=\sum_{\mathbf{x}}\log p\left(  \mathbf{x}\right)
=-c^{T}\mathbf{x}-\sum_{p=1}^{l_{i}}\log\sum_{h_{p}}e^{h_{p}\left(
b_{p}+W_{p}\mathbf{x}\right)  }%
\]
Here $\sum_{z}p\left(  \mathbf{x}\right)  \frac{\partial\digamma(\mathbf{x}%
)}{\partial\Lambda}$ is estimated by the Monte Carlo sampling,%
\begin{equation}
\sum_{z}p\left(  \mathbf{x}\right)  \frac{\partial\digamma(\mathbf{x}%
)}{\partial\Lambda}\thickapprox\frac{1}{s}\sum_{z\in S}\frac{\partial
\digamma(\mathbf{x})}{\partial\Lambda} \label{mc}%
\end{equation}
The above algorithm is for one RBM. The training process of multiple RBMs is
shown in Figure \ref{f2}. After the first model is trained, their weights are
fixed. The codes or hidden representations of the first model are sent to the
second model. The second model is trained using input $h_{1}\left(  k\right)
\in\Re^{l_{1}}$ and it generates its hidden representations $h_{2}\left(
k\right)  \in\Re^{l_{2}},$ which is the input of the third model.%

%TCIMACRO{\FRAME{ftbpFU}{2.2225in}{2.5106in}{0pt}{\Qcb{Cascade traning of
%multipe RBMs}}{\Qlb{f2}}{f2.wmf}{\special{ language "Scientific Word";
%type "GRAPHIC";  maintain-aspect-ratio TRUE;  display "USEDEF";
%valid_file "F";  width 2.2225in;  height 2.5106in;  depth 0pt;
%original-width 4.3889in;  original-height 4.9649in;  cropleft "0";
%croptop "1";  cropright "1";  cropbottom "0";
%filename 'fnn2.eps';file-properties "XNPEU";}} }%
%BeginExpansion
\begin{figure}
[ptb]
\begin{center}
\includegraphics[
height=2.5106in,
width=2.2225in
]%
{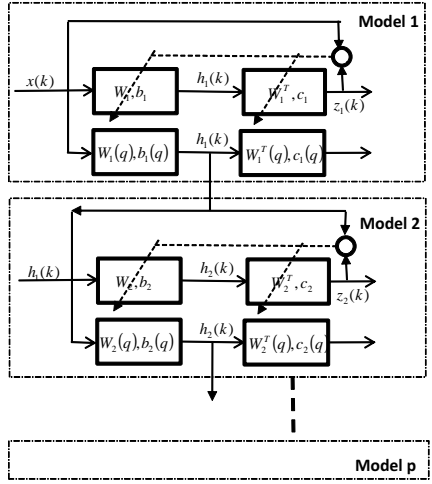}%
\caption{Cascade traning of multipe RBMs}%
\label{f2}%
\end{center}
\end{figure}
%EndExpansion

If $\mathbf{x}$ is encoded into binary values as (\ref{bin}), \textit{i.e.,
}$x_{t}\in\left\{  0,1\right\}  ,$ the binary hidden units are\textit{\ }%
$h_{p}\in\left\{  0,1\right\}  ,$%
\begin{equation}%
\begin{array}
[c]{l}%
p\left(  h_{p}=1|\mathbf{x}\right)  _{p=1\cdots l_{i}}=\phi\left[
W_{p}\mathbf{x}+b_{p}\right] \\
p\left(  x_{t}=1|h\right)  _{t=1\cdots l_{i-1}}=\phi\left[  W_{t}^{T}%
h+c_{t}\right]
\end{array}
\label{pnn}%
\end{equation}

If the input $\mathbf{x}$ uses continuous value. We first normalize
$\mathbf{x}$ in $[0,1].$ The conditional probability for the $j$-th visible
node is%
\[
P(x_{j}|h)=\frac{e^{(V_{j}^{T}\overline{h}+c_{j})x_{j}}}{\int_{\widehat{x}%
_{j}}e^{(V_{j}^{T}\overline{h}+c_{j})\widehat{x}_{j}}d\widehat{x}_{j}}%
\]
Since $x_{j}\in\lbrack0,1],$ the probability distribution is $P(x_{j}%
|h)=\frac{a_{j}e^{a_{j}x_{j}}}{e^{a_{j}}-1},$ $a_{j}=V_{j}^{T}\overline
{h}+c_{j}$. The accumulative conditional probability from sampling process is
\begin{equation}
P_{C}(x_{j}|h)=\frac{e^{a_{j}x_{j}}-1}{e^{a_{j}}-1} \label{cpd2}%
\end{equation}
The expected value of the distribution is $E[x_{j}]=\frac{1}{1-e^{-a_{j}}%
}-\frac{1}{a_{j}}$.

\subsection{Joint distribution $p(\mathbf{x},y)$}

Here we will discuss how to maximize the joint probability $J_{c}(D)=\sum
_{D}\log p(\mathbf{x,}y)$ by\ RBM training. The parameters of RBMs are updated
as
\begin{equation}
\Lambda\left(  k+1\right)  =\Lambda\left(  k\right)  -\eta_{2}\frac
{\partial\left[  \log p(\mathbf{x,}y)\right]  }{\partial\Lambda} \label{up1}%
\end{equation}
where $\eta_{2}$ is the learning rate. By the chain rule $p(\mathbf{x}%
,y,h)=p(h|\mathbf{x},y)p(\mathbf{x,}y)$,%
\begin{equation}
\frac{\partial\log p(\mathbf{x,}y)}{\partial\Lambda}=E_{(\mathbf{x}%
,y,h)}\left[  \frac{\partial E(\mathbf{x},y,h)}{\partial\Lambda}\right]
-E_{(h|\mathbf{x},y)}\left[  \frac{\partial E\left(  \mathbf{x},y,h\right)
}{\partial\Lambda}\right]  \label{upj}%
\end{equation}
\bigskip where $h$ is the hidden variable to capture the relationship between
$\mathbf{x}$ and $y$.

The expectations $E_{(\mathbf{x},y,h)}$ and $E_{(h|\mathbf{x},y)}$ cannot be
computed directly. The contrastive divergence (CD) approach \cite{Hinton} is
applied in this paper. From the starting point $\left[  \mathbf{x}%
(k),y(k)\right]  $, we sample the hidden state $h$ using this trigger. The
conditional probability $p(h|\mathbf{x},y)$ uses the sampling processes of
$p(\mathbf{x|h})$ and $p(\mathbf{y|h}).$ This process is repeated $\kappa$
times. These conditional probabilities are%
\begin{equation}%
\begin{array}
[c]{l}%
p(h|\mathbf{x},y)=%
%TCIMACRO{\dprod \limits_{j}}%
%BeginExpansion
{\displaystyle\prod\limits_{j}}
%EndExpansion
p(h_{j}|\mathbf{x},y)\\
p(\mathbf{x|}h)=%
%TCIMACRO{\dprod \limits_{\kappa}}%
%BeginExpansion
{\displaystyle\prod\limits_{\kappa}}
%EndExpansion
p(x_{i}|h)\\
p(y\mathbf{|}h)=%
%TCIMACRO{\dprod \limits_{\kappa}}%
%BeginExpansion
{\displaystyle\prod\limits_{\kappa}}
%EndExpansion
\left(  y_{i}|h\right)
\end{array}
\label{xyhprob}%
\end{equation}
which are used as the intermediate results in CD.

If $\mathbf{x}$ and $y$ are encoded into binary values as (\ref{bin}),%
\begin{equation}%
\begin{array}
[c]{l}%
p(h_{j}=1|\mathbf{x},\mathbf{y})=\text{sign}\left(  c_{j}+\sum_{\kappa
}V_{j\kappa}y_{\kappa}+\sum_{i}W_{ji}x_{i}\right) \\
p(x_{i}=1|\mathbf{h})=\text{sign}\left(  b_{i}+\sum_{j}W_{ji}h_{j}\right) \\
p(y_{\kappa}=1|\mathbf{h})=\text{sign}\left(  d_{\kappa}+\sum_{j}V_{j\kappa
}h_{j}\right)
\end{array}
\label{indProb}%
\end{equation}
The binary encoding method causes dramatically large training data. The
dimension of $\mathbf{x}(k)$ increase from $n$ to $2^{n\times r}$.

If $\mathbf{x}$ and $y$ are used as continuous values. The above three
conditional probabilities are calculated as follows.

1) The conditional probability of $\mathbf{x}\left(  k\right)  $ given $h.$
\[
p(\mathbf{x}|h)=\frac{p(\mathbf{x},h)}{p(h)}=\frac{\int_{\bar{y}}%
p(\mathbf{x},h,\bar{y})d\bar{y}}{\int_{\bar{y}}\int_{\mathbf{\bar{x}}%
}p(\mathbf{\bar{x}},h,\bar{y})d\mathbf{\bar{x}}d\bar{y}}=\frac{e^{h^{T}%
W\mathbf{x}+b^{T}\mathbf{x}}}{\int_{\mathbf{\bar{x}}}e^{h^{T}W\mathbf{\bar{x}%
}+b^{T}\mathbf{\bar{x}}}d\mathbf{\bar{x}}}=\prod_{i}p(\bar{x}_{i}|h)
\]
where $\bar{h}$, $\bar{y}$ and $\mathbf{\bar{x}}$ denote the silent variables
of $h$, $y$ and $\mathbf{x.}$%
\begin{equation}
p(x_{i}|h)=e^{x_{i}\left(  b_{i}+\sum_{j}w_{ji}h_{j}\right)  }/\int_{\bar
{x}_{i}}e^{\bar{x}_{i}\left(  b_{i}+\sum_{j}w_{ji}h_{j}\right)  }d\bar{x}_{i}
\label{xprobhconesp}%
\end{equation}
We explore three different cases for the domain of $x_{i}$: $\left[
0,\infty\right)  ,$ $\left[  0,1\right]  $ and $\left[  -\delta,\delta\right]
$ where $\delta\in\Re^{+}.$ For the case of $x_{i}\in$ $\left[  0,\infty
\right)  ,$ if we define $\alpha_{i}(h)=b_{i}+\sum_{j}w_{ji}h_{j},$ we can
directly evaluate the integral taking into account that $\alpha_{i}(h)<0,$
$\forall h.$ In order to ensure that the integral converges, the evaluations
of three integrals are presented in Table 1.%

\begin{align*}
&  \text{Table 1.- Probability expressions for }p(x|h)\\
&
\begin{tabular}
[c]{|l|l|l|l|}\hline
$\text{Interval}$ & $\left[  0,\infty\right)  $ & $\left[  0,1\right]  $ &
$\left[  -\delta,\delta\right]  $\\\hline
$p(x_{i}|h)$ & $-\alpha_{i}e^{\alpha_{ix_{i}i}}$ & $\frac{\alpha_{i}%
e^{\alpha_{i}x_{i}}}{e^{\alpha_{i}}-1}$ & $\frac{\alpha_{i}e^{\alpha_{i}x_{i}%
}}{e^{\delta\alpha_{i}}-e^{-\delta\alpha_{i}}}$\\\hline
$P_{c}(x_{i}|h)$ & $1-e^{\alpha_{i}x_{i}}$ & $\frac{e^{\alpha_{i}x_{i}}%
-1}{e^{\alpha_{i}}-1}$ & $\frac{e^{\alpha_{i}\bar{x}_{i}}-e^{-\delta\alpha
_{i}}}{e^{\delta\alpha_{i}}-e^{-\delta\alpha_{i}}}$\\\hline
$E\left[  x_{i}h\right]  $ & $-\frac{1}{\alpha_{i}}$ & $\frac{1}%
{1-e^{-\alpha_{i}}}-\frac{1}{\alpha_{i}}$ & $\delta\frac{e^{\delta\alpha_{i}%
}+e^{-\delta\alpha_{i}}}{e^{\delta\alpha_{i}}-e^{-\delta\alpha_{i}}}-\frac
{1}{\alpha_{i}}$\\\hline
\end{tabular}
\end{align*}

2) Probability of $y$ given $h.$%
\begin{equation}
p(y|h)=\frac{p(y,h)}{p(h)}=\frac{\int_{\bar{x}_{i}}p(\mathbf{x},h,\bar
{y})d\mathbf{\bar{x}}}{\int_{\bar{y}}\int_{\mathbf{\bar{x}}}p(\mathbf{\bar{x}%
},h,\bar{y})d\mathbf{\bar{x}}d\bar{y}}==\frac{e^{h^{T}Vy+D^{T}y}}{\int
_{\bar{y}}e^{h^{T}V\bar{y}+d^{T}\bar{y}}d\bar{y}}=\prod_{\kappa}p(y_{\kappa
}|h) \label{yprobhcont}%
\end{equation}
where%
\begin{equation}
p(y|h)=\frac{e^{(h^{T}V+D)y}}{\int_{\bar{y}}e^{(h^{T}V^{\prime}+D)\bar{y}%
}d\bar{y}} \label{yprobhcontesp}%
\end{equation}
If we define $\gamma(h)=h^{T}V+D,$ the evaluations of $p(y|h)$ are presented
in Table 2.%

\begin{align*}
&  \text{Table 2.}\ \text{Probability expressions for }p(y|h)\\
&
\begin{tabular}
[c]{|l|l|l|l|}\hline
$\text{Interval}$ & $\left[  0,\infty\right)  $ & $\left[  0,1\right]  $ &
$\left[  -\delta,\delta\right]  $\\\hline
$p(y|h)$ & $-\gamma e^{\gamma y}$ & $\frac{\gamma e^{\gamma y}}{e^{\gamma}-1}$
& $\frac{\gamma e^{\gamma y}}{e^{\delta\gamma}-e^{-\delta\gamma}}$\\\hline
$P_{c}(y|h)$ & $1-e^{\gamma y}$ & $\frac{e^{\gamma y}-1}{e^{\gamma}-1}$ &
$\frac{e^{\gamma y}-e^{-\delta\gamma}}{e^{\delta\gamma}-e^{-\delta\gamma}}%
$\\\hline
$E\left[  y|h\right]  $ & $-\frac{1}{\gamma}$ & $\frac{1}{1-e^{-\gamma}}%
-\frac{1}{\gamma}$ & $\delta\frac{e^{\delta\gamma}+e^{-\delta\gamma}%
}{e^{\delta\gamma}-e^{-\delta\gamma}}-\frac{1}{\gamma}$\\\hline
\end{tabular}
\end{align*}
\bigskip

3) Probability of $h$ given $\mathbf{x}$ and $y.$ The hidden units $h_{j}$
have binary values, while $\mathbf{x}$ and $y$ have continuous values. So
$p(h|\mathbf{x},y)$ is%
\begin{equation}
p(h|\mathbf{x},y)=\frac{p(\mathbf{x,}y\mathbf{,}h)}{p(\mathbf{x,}y)}=\prod
_{j}p(h_{j}|\mathbf{x,}y)\nonumber
\end{equation}
where $v_{j}$ denotes the $j$-th element of the vector $V$.

To calculate the gradient (\ref{upj}), we use following modified CD algorithm.

\textit{Algorithm 1:}

1.- Take a training pair $x(k),y(k)$

2.- Initialize $x_{1}=x(k)$ and$,y_{1}=y(k)$

3.- Sample $h_{1}$ from $p(h|x_{1},y_{1})$

4.- Sample $x_{2}$ and$,y_{2}$ from $p(x|h_{1})$ and $p(y|h_{1})$

5.- Sample $h_{2}$ from $p(h|x_{2},y_{2})$

This algorithm replaces the expectations with the $1$- steps Gibbs sampling
process. This process is initiated by pre-training of $p(\mathbf{x})$ as the
initial weights.

\section{RBM training for conditional probability}

The\ unsupervised pre-training (\ref{up3}) generates the probability
distribution and extracts the features of the input $\mathbf{x}\left(
k\right)  $. The supervised learning (\ref{up1}) can only obtain the
sub-optimal model for the system identification index (\ref{maxJ1}). However,
the training of (\ref{up1}) discussed in the above is relatively simple.

Now we use the supervised learning to obtain the conditional distribution
$p\left(  y|\mathbf{x}\right)  $ via RBM training. The parameters of RBMs are
updated as
\begin{equation}
\Lambda\left(  k+1\right)  =\Lambda\left(  k\right)  -\eta_{3}\frac
{\partial\left[  \log p\left(  y|\mathbf{x}\right)  \right]  }{\partial
\Lambda} \label{up2}%
\end{equation}
where $\eta_{3}>0$ is the training factor, $\Lambda=\left\{
W,b,c,d,V\right\}  ,$ $\frac{\partial\log p\left(  y|\mathbf{x}\right)
}{\partial\Lambda}$ will be calculated as follows. This is the final goal of
the nonlinear system modeling (\ref{maxJ}). Because $\log p(y\mathbf{|x})=\log
p(\mathbf{x,}y)-\log p(\mathbf{x})$, from (\ref{engfun})
\begin{equation}%
\begin{array}
[c]{l}%
p\left(  y|\mathbf{x}\right)  =\frac{p(\mathbf{x,}y)}{p(\mathbf{x})}=%
%TCIMACRO{\dsum \limits_{h}}%
%BeginExpansion
{\displaystyle\sum\limits_{h}}
%EndExpansion
e^{-E(\mathbf{x},y,h)}/%
%TCIMACRO{\dsum \limits_{y,h}}%
%BeginExpansion
{\displaystyle\sum\limits_{y,h}}
%EndExpansion
e^{-E(\mathbf{x},y,h)}\\
\log p\left[  \mathbf{x}(k),y(k)\right]  =\log\sum_{h}e^{-E\left[
\mathbf{x}(k),y(k),h\right]  }-\log\sum_{y,h}e^{-E\left[  \mathbf{x}%
(k),y(k),h\right]  }%
\end{array}
\label{conditionalprob}%
\end{equation}
So the gradient of the negative log-likelihood with respect to the parameter
$\Lambda$ is%
\begin{equation}
-\frac{\partial\log p\left(  y|\mathbf{x}\right)  }{\partial\Lambda}%
=\frac{\sum_{h}e^{-E\left[  \mathbf{x},y,h\right]  }\frac{\partial E\left[
\mathbf{x},y,h\right]  }{\partial\Lambda}}{\sum_{h}e^{-E(\mathbf{x},y,h)}%
}-\frac{\sum_{y,h}e^{-E\left[  \mathbf{x},y,h\right]  }\frac{\partial E\left[
\mathbf{x},y,h\right]  }{\partial\Lambda}}{\sum_{y,h}e^{-E\left[
\mathbf{x},y,h\right]  }} \label{logPXYexpr}%
\end{equation}
In the form of mathematical expectation, the gradient is
\begin{equation}
-\frac{\partial\log p\left(  y|\mathbf{x}\right)  }{\partial\Lambda
}=E_{(h|\mathbf{x},y)}\left[  \frac{\partial E(\mathbf{x},y,h)}{\partial
\Lambda}\right]  -E_{(y,h|\mathbf{x})}\left[  \frac{\partial E(\mathbf{x}%
,y\mathbf{,}h)}{\partial\Lambda}\right]  \label{me}%
\end{equation}
Both probability expectations of (\ref{me}) can be computed using Gibbs
sampling and CD algorithm. The CD algorithm needs alternation sampling
processes over the distributions (\ref{xyhprob}). However, there are not
compact expressions for $p(h|\mathbf{x},y)$ and $p(y,h|\mathbf{x}).$ In order
to implement the CD\ algorithm for $E_{(y,h|\mathbf{x})}\left[  \frac{\partial
E(\mathbf{x},y\mathbf{,}h)}{\partial\Lambda}\right]  ,$\ we calculate
$p(y,h|\mathbf{x})$ as%
\begin{equation}
p(y,h|\mathbf{x})=\frac{e^{-E(\mathbf{x},y,h)}}{\sum_{y,h}e^{-E(\mathbf{x}%
,y,h)}}=\frac{e^{h^{T}W\mathbf{x}+b^{T}\mathbf{x}+c^{T}h+dy+h^{T}Vy}}%
{\sum_{y,h}e^{h^{T}W\mathbf{x}+b^{T}\mathbf{x}+c^{T}h+dy+h^{T}Vy}dy}
\label{yhprobxBin}%
\end{equation}
The calculation of (\ref{yhprobxBin}) is expensive, because it requires\ to
calculate more than $2^{m+l}$ possible values. It is tractable for system
identification. After this distribution is obtained, we just should sample all
possible values for $y$ and $h$. Once $p(y,h|\mathbf{x})$ is calculated, we
need $p(\mathbf{x}|y,h)$ to complete the Gibbs sampling as%
\begin{equation}
p(\mathbf{x}|y,h)=\frac{p(\mathbf{x,}y\mathbf{,}h)}{p(y\mathbf{,}h)}=\prod
_{i}p(\bar{x}_{i}|h) \label{xyh}%
\end{equation}
The calculation of $\frac{\partial\log p\left(  y|\mathbf{x}\right)
}{\partial\Lambda}$ given by (\ref{me}) cannot be done directly. We must use
again the CD algorithm. The term $E_{h|(\mathbf{x},y)}\left[  \frac{\partial
E(\mathbf{x},y,h)}{\partial\Lambda}\right]  $ can be computed as%
\begin{equation}%
\begin{array}
[c]{l}%
E_{(\mathbf{x},y,h)}\left[  \frac{\partial E(\mathbf{x},y,h)}{\partial\Lambda
}\right]  \approx\frac{\partial E(\mathbf{x}_{2},y_{2},h_{2})}{\partial
\Lambda}\\
E_{(h|\mathbf{x},y)}\left[  \frac{\partial E\left(  \mathbf{x},y,h\right)
}{\partial\Lambda}\right]  \approx\frac{\partial E(\mathbf{x},y,h_{2}%
)}{\partial\Lambda}%
\end{array}
\label{expexp}%
\end{equation}

In this paper, we modify the learning algorithm of RBM, such that nonlinear
system can be modeled by continuous values. In order to train the parameters
in (\ref{up2}), we need the three conditional probabilities discussed in the
above session, and the following three conditional probabilities.

4) Probability of $y$ given $\mathbf{x.}$ We assume that $y$ is a scalar, the
bias variable $d$ is a real number, the weight matrix $V$ is a real vector.%
\begin{equation}
p(y|\mathbf{x})=\frac{p(\mathbf{x,}y)}{p(\mathbf{x})}=\frac{\sum_{\bar{h}%
}p(\mathbf{x},y,\bar{h})}{\int_{\bar{y}}\sum_{\bar{h}}p(\mathbf{x},\bar
{y},\bar{h})d\bar{y}}=\frac{\sum_{\bar{h}}e^{\bar{h}^{T}W\mathbf{x}%
+b^{T}\mathbf{x}+c^{T}\bar{h}+dy+\bar{h}^{T}Vy}}{\int_{\bar{y}}\sum_{\bar{h}%
}e^{h^{T}W\mathbf{x}+b^{T}\mathbf{x}+c^{T}\bar{h}+d\bar{y}+\bar{h}^{T}V\bar
{y}}d\bar{y}} \label{yprobxcont}%
\end{equation}
Using Fubini's Theorem, the integral and the sum are%
\begin{equation}
p(y\mathbf{|x})=\frac{e^{dy}\prod_{j}\left(  1+e^{\tau_{j}(\mathbf{x,}%
y)}\right)  }{\int_{\bar{y}}e^{d\bar{y}}\prod_{j}\left(  1+e^{\tau
_{j}(\mathbf{x,}\bar{y})}\right)  d\bar{y}} \label{yprobxred}%
\end{equation}

5) Probability of $\mathbf{(}y,h\mathbf{)}$ given $\mathbf{x.}$ The second
term of the negative log-likelihood (\ref{logPXYexpr}) can be computed as%
\begin{equation}
p\left[  (y\mathbf{,}h)|\mathbf{x}(k)\right]  =\frac{e^{-E(\mathbf{x}%
(k),\bar{y},\bar{h})}}{\sum_{\bar{y},\bar{h}}e^{-E(\mathbf{x}(k),\bar{y}%
,\bar{h})}}=\frac{e^{h^{T}W\mathbf{x}(k)+b^{T}\mathbf{x}(k)+c^{T}h+dy+h^{T}%
Vy}}{\int_{\bar{y}}\sum_{\bar{h}}e^{h^{\prime T}W\mathbf{x}(k)+b^{T}%
\mathbf{x}(k)+c^{T}\bar{h}+d\bar{y}+h^{\prime T}V\bar{y}}d\bar{y}}
\label{yhprobx}%
\end{equation}
The integral of the denominator is expanded as $\int_{\bar{y}}e^{d\bar{y}%
}\prod_{j}\left(  1+e^{\tau_{j}(\mathbf{x}(k),\bar{y})}\right)  d\bar{y}$. In
order to find a closed form of the solution, we define $\digamma=\{\tau
_{1,}\tau_{2},...,\tau_{n}\}$, which is incomplete power set $P(\digamma)$,
because the empty set is not included in $P(\digamma)$. Associated with
$\digamma$ with elements $P(\digamma)=\{P_{\digamma1},P_{\digamma2},...\}$,
the elements $P_{\digamma i}$ contain all possible combinations of elements
$\tau_{j}$. The finite product $\prod_{j}\left(  1+e^{\tau_{j}(\mathbf{x}%
(k),\bar{y})}\right)  $ can be expressed as%
\begin{equation}
\prod_{j}\left(  1+e^{\tau_{j}}\right)  =1+\sum_{P_{\digamma i}}e^{\sum
\tau_{\gamma}}\text{ } \label{infProd}%
\end{equation}
where $\gamma$ is an index for $\tau,$ which takes values such that
$\tau_{\gamma}\in P_{\digamma i}$. The integral then becomes $\int_{\bar{y}%
}e^{d\bar{y}}\left(  1+\sum_{P_{\digamma i}}e^{\sum\tau_{\gamma}%
(\mathbf{x}(k),\bar{y})}\right)  d\bar{y}.$ Because $\tau_{j}=\sum_{i}%
w_{ji}\bar{x}_{i}+v_{j}y+c_{j},$ define the vector $w_{j}=\left[
w_{j1}...w_{jl}\right]  ,$ then $\tau_{j}=w_{j}x+v_{j}y+c_{j}$. Considering
the expression for $\tau_{j},$ the value of the integral is%
\begin{equation}
\int_{\bar{y}}\left(  e^{d\bar{y}}+\sum_{P_{\digamma i}}e^{\sum w_{\gamma
}\mathbf{x}(k)+c_{\gamma}}e^{\left(  d+\sum v_{\gamma}\right)  \bar{y}%
}\right)  d\bar{y} \label{denoExp2}%
\end{equation}
For the interval $[0,\infty),$%
\begin{equation}
-\frac{1}{D}-\sum_{P_{\digamma i}}\frac{1}{D+\sum v_{\gamma}}e^{\sum
w_{\gamma}\mathbf{x}(k)+c_{\gamma}} \label{int0inf}%
\end{equation}
For the interval $[0,1],$%
\begin{equation}
\frac{1}{d}\left(  e^{D}-1\right)  +\sum_{P_{\digamma i}}\frac{e^{\sum
w_{\gamma}\mathbf{x}(k)+c_{\gamma}}}{F+\sum v_{\gamma}}\left(  e^{D+\sum
v_{\gamma}}-1\right)  \label{int01}%
\end{equation}
For the interval $[-\delta,\delta],$%
\begin{equation}
\frac{1}{d}\left(  e^{D\delta}-e^{-D\delta}\right)  +\sum_{P_{\digamma i}%
}\frac{e^{\sum w_{\gamma}\mathbf{x}(k)+c_{\gamma}}}{D+\sum v_{\gamma}}\left(
e^{\left(  D+\sum v_{\gamma}\right)  \delta}-e^{-\left(  D+\sum v_{\gamma
}\right)  \delta}\right)  \label{intmdd}%
\end{equation}
The sum $\sum_{P_{\digamma i}}$ is performed along the elements of the power
set which is computational expensive. The number of elements is $2^{n},$ which
represents all possible combinations. For system identification, the number of
visible and hidden units is no so big, so the procedure becomes tractable.

6) Probability of $\mathbf{x}$ given $(y\mathbf{,}h).$ We have shown that
$p\left[  \mathbf{x}|(y\mathbf{,}h)\right]  =\prod_{i}p(x_{i}|h).$ In the
intervals $\left[  0,\infty\right)  ,$ $\left[  0,1\right]  $ and $\left[
-\delta,\delta\right]  ,$ we get the same expressions presented in Table 1 for
$p(\mathbf{x|}h)$.

Finally we use the following CD algorithm to calculate $E_{h|(\mathbf{x}%
,y)}\left[  \frac{\partial E(\mathbf{x},y,h)}{\partial\Lambda}\right]  $.

\textit{Algorithm 2}

1.- Take a training pair $x(k),y(k)$

2.- Initialize $x_{a}=x(k)$

3.- Sample $y_{a}$ and $h_{a}$ from $p\left[  (y\mathbf{,}h)|\mathbf{x}%
_{a}\right]  $

4.- Sample $x_{b}$ from $p\left[  \mathbf{x}|(y_{a}\mathbf{,}h_{a})\right]  $

5.- Sample $y_{b}$ and $h_{b}$ from $p\left[  (y\mathbf{,}h)|\mathbf{x}%
_{b}\right]  $

\section{Simulations}

In this section, we use two benchmark examples to show the effectiveness of
the conditional probability based system identification.

\textbf{Gas furnace process}

One of the most utilized benchmark examples in system identification is the
famous gas furnace \cite{Jenkins}. The air and methane are mixed to create gas
mixture which contains carbon dioxide. The control $u(k)$ is methane gas, the
output $y(k)$ is $CO_{2}$ concentration. This process is modeled as
(\ref{planta}). In this paper, we use the simulation model, \textit{i.e.},
only the input $u(k)$ is used to obtain the modeled output,
\begin{equation}
\hat{y}(k)=NN[u\left(  k\right)  ,\cdots u\left(  k-n_{u}\right)  ]^{T}
\label{sim}%
\end{equation}
This model is more difficult than the following prediction model, who uses
both the input and the past output $y(k),$
\begin{equation}
\hat{y}(k)=NN[y\left(  k-1\right)  ,\cdots y\left(  k-n_{y}\right)  ,u\left(
k\right)  ,\cdots u\left(  k-n_{u}\right)  ]^{T} \label{pre}%
\end{equation}
The big advantage of the simulation model (\ref{sim}) is the on-line
measurement of the gas furnace output is not needed.

The gas furnace are sampled continuously in $9$ second intervals. The data set
is composed of $296$ successive pairs of $[u(k),y(k)]$. $200$ samples are used
as the training data, the rest $96$ samples are for the testing. We compare
our conditional probability calculation with restricted Boltzmann machine
(RBM-C) and the joint distribution with restricted Boltzmann machine (RBM-J),
to the feedforward neural networks, multilayer perceptrons (MLP), and the
support vector machine (SVM). The MLP has the same structure as the RBMs,
\textit{i.e.}, the same hidden layer and hidden node numbers. The SVMs use
three types of kernel: polynomial kernel (SVM-P), radial basis
function\ kernel (SVM-R), linear kernel (SVM-L).

We use the random search method \cite{Bergstra1} to determine\ how many
RBMs\ we need. The search ranges of the RBM number $l$ is $10\geq l\geq2,$ the
hidden node number $p$ is $40\geq p\geq5.$ The random search results are,
$l=3$, $p=20$. Three RBMs are used and each hidden layer has $30$ nodes. The
following steps are applied for RBM training.

a) Binary encoding and decoding. We used two resolutions, $4$ bits and $8$
bits, for $\mathbf{x}$ and $y$. The resolutions of the input are
$4\times\left(  n_{y}+n_{u}\right)  $ and $8\times\left(  n_{y}+n_{u}\right)
.$ The output data are sampled from $p(y|\mathbf{x})$ and decoded to
continuous equivalent values.

b) Probability distribution with continuous values. The data are normalized
into the interval $[0,1]$,
\begin{equation}
\mathbf{x}\left(  k\right)  =\frac{\mathbf{x}\left(  k\right)  -\min
_{k}\left\{  \mathbf{x}\left(  k\right)  \right\}  }{\max\left\{
\mathbf{x}\left(  k\right)  \right\}  -\min_{k}\left\{  \mathbf{x}\left(
k\right)  \right\}  },y=\frac{y-y_{\max}}{y_{\min}-y_{\max}} \label{normi}%
\end{equation}

c) Training: the RBM is trained using the coded data or continuous values. The
learning rates are $\eta_{1}=\eta_{2}=\eta_{3}=0.01$. Stochastic gradient
descents (\ref{up3}), (\ref{up1}), (\ref{up2}) are applied over the data set.
The algorithm use $10$ training epochs.

If we use the prediction model (\ref{pre}) to describe the gas furnace
process, with $n_{y}=1$ and $n_{u}=5,$ MLP, SVM, RBM-J and RBM-C work well.
The mean squared errors (MSE) of the testing data are similar small. In order
to show the noise resistance of the conditional probability methods, we added
noises to the data set,%
\begin{equation}
\mathbf{x}(k)=\mathbf{x}(k)+\mathbf{z}(k) \label{Noisy}%
\end{equation}
where $\mathbf{z}(k)$ is a normal distribution with $0$ average and $0.1$
standard deviation. The testing errors are shown in Table 3.
\begin{align*}
&  \text{Table 3. MSE of different prediction models with noise}\left(
\times10^{-3}\right) \\
&
\begin{tabular}
[c]{|c|c|c|c|c|c|}\hline
MLP & SVM-L & SVM-P & SVM-R & RBM-J & RBM-C\\\hline
$30.03$ & $23.01$ & $26.7$ & $20.70$ & $11.12$ & $8.05$\\\hline
\end{tabular}
\end{align*}

The probabilistic models have great advantage over MLP and SVM with respect to
noises and disturbances. The main reason is that we model the probability
distributions of the input and output, the noises and outliers in the data do
not affect the conditional distributions significantly.

Now we use the simulation model (\ref{sim}) to compare the above models. When
$n_{u}=5,$ MLP and SVM cannot model the process, RBM can but the MSE of the
testing error is about $100\times10^{-3}$. When $n_{u}=15,$ all models work.
The testing errors are shown in Table 4.%

\begin{align*}
&  \text{Table 4. MSE of different simulation models without noise }\left(
\times10^{-3}\right) \\
&
\begin{tabular}
[c]{|c|c|c|c|c|c|}\hline
MLP & SVM-L & SVM-P & SVM-R & RBM-J & RBM-C\\\hline
$81.7$ & $56.3$ & $65.2$ & $62.1$ & $26.32$ & $19.52$\\\hline
\end{tabular}
\end{align*}
Besides the supervised learning, RBMs can extract the input features with the
unsupervised learning. So RBMs work better than MLP and SVM when the input
vector $\mathbf{x}(k)$ does not include previous output $y\left(  k-1\right)
.$

To show the effectiveness of the deep structure and the binary encoding
method, we use $l$ RBMs. The testing errors are given in Table 5.%

\begin{align*}
&  \text{Table 5. MSE of different RBMs-C }\left(  \times10^{-3}\right) \\
&
\begin{tabular}
[c]{|l|l|l|l|l|}\hline
MSE & $1$ RBM & $2$ RBMs & $3$ RBMs & $4$ RBMs\\\hline
4 bits & $87.62$ & $32.36$ & $21.2$ & $31.67$\\\hline
8 bits & $74.61$ & $29.21$ & $19.5$ & $25.75$\\\hline
\end{tabular}
\end{align*}

By adding new feature extraction block, the MSE drops significantly. If the
number\ of RBM is more than $3,$ the MSE becomes worse. This means that it is
not necessary to add new RBM to extract information.

Both $4$ bits and $8$ bits encoding have good approximation results, see
Figure \ref{FigOutGas}. The high precise encoding helps to improve the model
accuracy. However, adding one bit in the encoding procedure immediately
doubles the computation time.%

%TCIMACRO{\FRAME{ftbpFU}{4.1104in}{2.4242in}{0pt}{\Qcb{Modeling the gas furnace
%using conditional probability and binary encoding.}}{\Qlb{FigOutGas}%
%}{digital.eps}{\special{ language "Scientific Word";  type "GRAPHIC";
%display "USEDEF";  valid_file "F";  width 4.1104in;  height 2.4242in;
%depth 0pt;  original-width 8.003in;  original-height 5.9949in;  cropleft "0";
%croptop "1";  cropright "1";  cropbottom "0";
%filename 'Digital.eps';file-properties "XNPEU";}} }%
%BeginExpansion
\begin{figure}
[ptb]
\begin{center}
\includegraphics[
height=2.4242in,
width=4.1104in
]%
{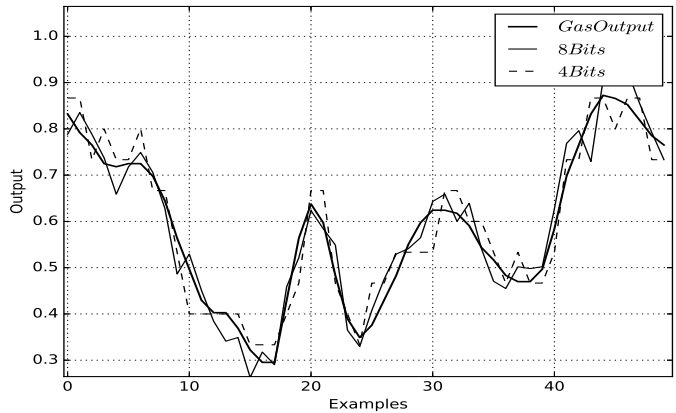}%
\caption{Modeling the gas furnace using conditional probability and binary
encoding.}%
\label{FigOutGas}%
\end{center}
\end{figure}
%EndExpansion

Finally, we test the continuous valued algorithm. We use the simulation model
(\ref{sim}). The training parameters are $k_{G}=1,$ $\eta_{1}=0.01,$ and $100$
training epochs for each RBM. For the output layer, the coded features have
$k_{G}=1,$ and $10$ training epochs. The testing MSE is $12.5\times10^{-3}.$
It is better than binary encoding, because it provides more information on
real axis.

\textbf{Wiener-Hammerstein system}

Wiener-Hammerstein system \cite{Schoukens} has a static nonlinear part
surrounded by two dynamic linear systems. There are $188,000$ input/output
pairs, defined $u(k)$ and $y\left(  k\right)  $. We use $100,000$ samples for
training, $88,000$ samples for testing. This is a big data modelling problem.

We use the simulation model (\ref{sim}) to compare these models, with
$n_{u}=15,$ $n_{y}=0.$ We only show the feature extraction of the proposed
method. The random search method uses the RBM number as $20\geq l\geq2,$ and
the hidden node number as $70\geq p\geq10.$ Then RBM number $l=5,$ the hidden
node number $p=50.$ The RBMs are trained using the Gibbs sampling, $k_{G}=1$,
the learning rates are $0.01$. It has $10$ training epochs. We use the same
structure for the MLP, five hidden layers, each layer has $50$ nodes.

Table 6 shows the testing errors of different methods for the simulation model
(\ref{sim}).%

\[%
\begin{array}
[c]{c}%
\text{Table 6. MSE of different simulation models }\left(  \times
10^{-3}\right) \\%
\begin{tabular}
[c]{|c|c|c|c|c|}\hline
MLP & SVM-L & SVM-P & SVM-R & RBM-C\\\hline
$56.03$ & $43.01$ & $48.01$ & $35.71$ & $12.70$\\\hline
\end{tabular}
\end{array}
\]
We can see that even for the best result obtained by the SVM-R, RBM-C is much
better than the others when the input to the models is only control $u(k).$

Now we show how does the training data seize affect the conditional
probability modelling. The stochastic gradient descent is a batch process. The
$100,000$ training samples are divided into several packages. All packages
have the same size. The package seizes are selected as $500,$ $1,000$ and
$5,000.$ The probability distributions are calculated with: binary encoding,
in the interval $[0,\infty),$ in the interval $[0,1),$ and in the interval
$[-\delta,\delta].$

We can see that the RBMs cannot model the probability distribution properly
with small batch number, for example $60$ batches. The training error with the
size $500$ is little bigger than the size $1,000.$ The fluctuations of the
training error with the size $5,000$ vanish, and the interval $[-\delta
,\delta]$ becomes unstable. So large batch seize may affect the distributions
by some mislead samples.

\section{Conclusions}

In this paper, the conditional probability based method is applied for
nonlinear system modelling. We show that this method is better than the other
data based models when there are noises and the previous outputs are not
available on-line. The modelling accuracy is satisfied with the binary
encoding and continuous values by the modified RBMs. The training algorithms
are obtained from maximizing the conditional likelihood of the data set. Two
nonlinear modelling problems are used to validate the proposed methods. The
simulation results show that the modelling accuracies are improved a lot when
there are noises and the simulation models have to be used.

\end{document}